\documentclass[12pt,letterpaper]{article}
\usepackage{amsmath,amsfonts,amsthm,amssymb,stmaryrd,bm,cite} 
\usepackage{relsize,tabls}
\allowdisplaybreaks  

\theoremstyle{plain}

\numberwithin{equation}{section}
\newtheorem{thm}{Theorem}[section]

\newenvironment{exam}[1]
{\begin{flushleft}\textbf{Example #1}.\enspace}%
{\end{flushleft}}

\newcommand{\positive}{{\mathbb N}}
\newcommand{\complex}{{\mathbb C}}

\newcommand{\lscript}{{\mathcal L}}

\newcommand{\rmprob}{\mathrm{prob}}

\newcommand{\vhat}{\widehat{v}}
\newcommand{\sixhat}{\widehat{6}}
\newcommand{\sevenhat}{\widehat{7}}
\newcommand{\eighthat}{\widehat{8}}

\newcommand{\ctimes}{\mathrel{\mathlarger\cdot}}

\newcommand{\ab}[1]{\left|#1\right|}
\newcommand{\brac}[1]{\left\{#1\right\}}
\newcommand{\paren}[1]{\left(#1\right)}
\newcommand{\sqbrac}[1]{\left[#1\right]}
\newcommand{\elbows}[1]{{\left\langle#1\right\rangle}}

\begin{document}

\title{WAVE EQUATIONS FOR\\DISCRETE QUANTUM GRAVITY
}
\author{S. Gudder\\ Department of Mathematics\\
University of Denver\\ Denver, Colorado 80208, U.S.A.\\
sgudder@du.edu
}
\date{}
\maketitle

\begin{abstract}
This article is based on the covariant causal set ($c$-causet) approach to discrete quantum gravity. A $c$-causet $x$ is a finite partially ordered set that has a unique labeling of its vertices. A rate of change on $x$ is described by a covariant difference operator and this operator acting on a wave function forms the left side of the wave equation. The right side is given by an energy term acting on the wave function. Solutions to the wave equation corresponding to certain pairs of paths in $x$ are added and normalized to form a unique state. The modulus squared of the state gives probabilities that a pair of interacting particles is at various locations given by pairs of vertices in
$x$. We illustrate this model for a few of the simplest nontrivial examples of $c$-causets. Three forces are considered, the attractive and repulsive electric forces and the strong nuclear force. Large models get much more complicated and will probably require a computer to analyze.
\end{abstract}

\section{The Basic Model}  
In a previous article, the author introduced a free wave equation for discrete quantum gravity \cite{gud15}. According to this model, the geometry of a universe is determined by nature (through a quantum formalism) and it is this geometry that we call gravity. We say that a particle is free if it is not acted on by any forces except gravity and gravity is not a force anyway, it is geometry. The propagation amplitude of a free particle is governed by the free wave equation. We showed in \cite{gud15} that solutions to this equation predict that particles tend to move along geodesics given by the geometry This implies that particles are attracted toward regions of large curvature. In this sense, the mass-energy distribution is determined by the geometry which is opposite to the postulates of classical general relativity theory.

In the present article, we extend the work in \cite{gud15} to include other interactions. Unlike the free wave equation which describes the motion of a single particle, the equations we introduce here describe two interacting particles. The main difference is that the free wave equation involves functions of a single variable, while the interaction wave equations involve functions of two variables. We now briefly review the basic model for this theory. For more details, we refer the reader to \cite{gud13,gud14,gud15}.

We shall model the structure of a universe by a causal set or causet \cite{bdghs, hen09,sor03,sur11,wc15}. Mathematically, a
\textit{causet} is a partially ordered set $(x,<)$. For $a,b\in x$, we interpret $a<b$ as meaning that $b$ is in the causal future of $a$. If $a<b$ and there is no $c\in x$ with $a<c<b$ we say that $a$ is a \textit{parent} of $b$ and write $a\prec b$. Denoting the cardinality of $x$ by
$\ab{x}$, a \textit{labeling} of $x$ is a map $\ell\colon x\to\brac{1,2,\ldots ,\ab{x}}$ such that $a<b$ implies $\ell (a)<\ell (b)$. A labeling of $x$ may be considered a ``birth order'' of the vertices of $x$. A \textit{covariant causet} ($c$-\textit{causet}) is a causet that has a unique labeling \cite{gud13,gud14,gud15}. In this article, we shall only model possible universes by $c$-causets.

A \textit{path} in a $c$-causet $x$ is a finite sequence $a_1a_2\cdots a_n$ with $a_1\prec a_2\prec\cdots\prec a_n$. The
\textit{height} $h(a)$ of $a\in x$ is the cardinality, minus one, of a longest path in $x$ that ends with $a$. Two vertices $a,b$ are
\textit{comparable} if $a<b$ or $b<a$. It is shown in \cite{gud14} that a causet $x$ is a $c$-causet if and only if $a,b\in x$ are comparable whenever $h(a)\ne h(b)$. We call the set
\begin{equation*}
S_j(x)=\brac{a\in x\colon h(a)=j}
\end{equation*}
the $j$th \textit{shell} of $x$. Letting $s_j(x)=\ab{S_j(x)}$, $j=0,1,2,\ldots ,k$, we call $\paren{s_0(x),s_1(x),\ldots ,s_k(x)}$ the
\textit{shell sequence} of $x$. A $c$-causet is uniquely determined by its shell sequence. Conversely, any finite sequence of positive integers is the shell sequence for a unique $c$-causet.

Let $\omega =a_1a_2\cdots a_n$ be a path in $x$ where $a_j\in\positive$ are the labels of the vertices. The \textit{length} of $\omega$ is
\begin{equation*}
\lscript (\omega )=\sqbrac{\sum _{j=1}^{n-1}(a_{j+1}-a_j)^2}^{1/2}
\end{equation*}
A \textit{geodesic} from $a$ to $b$ where $a<b$, is a path from $a$ to $b$ of smallest length. For $a<b$, we define the \textit{distance from}
$a$ \textit{to} $b$ to be $d(a,b)=\lscript (\omega )$ where $\omega$ is a geodesic from $a$ to $b$. If $a$ and $b$ are incomparable then
$a$ and $b$ are in the same shell and we define $d(a,b)=\ab{\ell (a)-\ell (b)}$. It can be shown that if $a<b<c$ then
$d(a,c)\le d(a,b)+d(b,c)$ \cite{gud14}. Moreover, it is clear that if $a,b,c\in S_j(x)$ for some $j$, then again, $d(a,c)\le d(a,b)+d(b,c)$. Even if we extend the definition of $d$ by defining $d(b,a)=d(a,b)$ when $a<b$, the triangle inequality does not hold in general.

The shell sequence determines the ``shape'' or geometry of a $c$-causet $x$. We view $x$ as a framework or scaffolding of a possible universe. The vertices represent tiny cells that may or may not be occupied by a particle. This geometry gives the kinematics of the system. The dynamics is described in terms of paths in $x$. A free particle tends to move along a geodesic \cite{gud15}, while interacting particles may be forced to propagate along other paths. The remainder of this article describes how this propagation is governed by discrete wave equations. We have already considered gravitational wave equations in \cite{gud15}. Section~2 presents a force-free wave equation and Section~3 studies electric force wave equations. Finally, Section~4 considers the strong nuclear force. Many examples that illustrate this theory are given for the simplest nontrivial $c$-causet.

\section{Difference Operators} 
Let $\omega =a_1a_2\cdots a_n$ and $\omega '=b_1b_2\cdots b_n$ be two paths with the same length in a $c$-causet $x$. We say that
$\omega ,\omega '$ are \textit{equal-time} paths if $h(a_1)=h(b_1)$ from which it follows that $h(a_j)=h(b_j)$, $j=1,2,\ldots ,n$.
Two equal-time paths are \textit{noncrossing} if $a_j\ne b_j$, $j=1,\ldots ,n$. The reason that we shall only consider noncrossing paths when they are equal-time is that we need the condition $d(a_j,b_j)\ne 0$. Notice that this condition automatically holds if $\omega ,\omega '$ are not equal-time. If $v$ is a complex-valued function of two variables whose domain includes $\brac{(a_j,b_j)\colon j=1,2,\ldots ,n}$, we define the
\textit{covariant difference operator} on $v$ by
\begin{equation}         
\label{eq21}
\nabla _{\omega ,\omega '}v(a_j,b_j)=d(a_{j-1},b_{j-1})v(a_j,b_j)-d(a_j,b_j)v(a_{j-1},b_{j-1})
\end{equation}
$j=2,3,\ldots ,n$.

We only consider $\nabla _{\omega ,\omega '}$ for noncrossing paths when $\omega ,\omega '$ are equal-time. Notice that
$\nabla _{\omega ,\omega '}d(a_j,b_j)=0$, $j=2,3,\ldots ,n$, which is why we call $\nabla _{\omega ,\omega '}$, the
\textit{covariant} difference operator. We think of the function $v$ in \eqref{eq21} as describing amplitudes of a pair of particles moving along the paths $\omega ,\omega '$. We shall usually assume that the two particles are indistinguishable because it makes the theory slightly simpler, but this assumption can be dropped. We call $v$ a \textit{force-free wave function} for $(\omega ,\omega ')$ if
\begin{equation}         
\label{eq22}
\nabla _{\omega ,\omega '}v(a_j,b_j)=0
\end{equation}
$j=2,3,\ldots ,n$. We use this terminology because such a $v$ is supposed to describe a pair of particles that are not subject to forces. This includes gravity which is not really a force because gravity is given by the geometry of the $c$-causet. But all particles are subject to gravity so this is actually an idealized approximation in which gravity is neglected. More realistic situations are treated in the next two sections where we include force terms.

\begin{thm}       
\label{thm21}
A function $v(a_j,b_j)$, $j=1,2,\ldots ,n$, satisfies \eqref{eq22} if and only if $v(a_j,b_j)=cd(a_j,b_j)$, $j=2,3,\ldots ,n$ for some fixed
$c\in\complex$.
\end{thm}
\begin{proof}
If $v$ satisfies \eqref{eq22} then by \eqref{eq21} we have
\begin{equation*}
v(a_j,b_j)=\frac{d(a_j,b_j)}{d(a_{j-1},b_{j-1})}\,v(a_{j-1},b_{j-1})\quad j=2,3,\ldots ,n
\end{equation*}
If $j\ge 3$ we can continue to obtain
\begin{align*}
v(a_j,b_j)&=\frac{d(a_j,b_j)}{d(a_{j-1},b_{j-1})}\,\frac{d(a_{j-1},b_{j-1})}{d(a_{j-2},b_{j-2})}\,v(a_{j-2},b_{j-2})\\\noalign{\medskip}
  &=\frac{d(a_j,b_j)}{d(a_{j-2},b_{j-2})}\,v(a_{j-2},b_{j-2})
\end{align*}
Continue this process to obtain
\begin{equation*}
v(a_j,b_j)=\frac{d(a_j,b_j)}{d(a_1,b_1)}\,v(a_1,b_1)\end{equation*}
for $j=1,2,\ldots ,n$. Letting $c=v(a_1,b_1)/d(a_1,b_1)$ gives our result. The converse clearly holds.
\end{proof}

We shall mainly consider equal-time noncrossing paths $\omega ,\omega '$ in the sequel because they give a simpler quantum theory. Moreover, in this case $d(a_j,b_j)=\ab{a_j-b_j}$, where we represent a vertex by its label, which is a further simplification. In this situation, let $v$ be a force-free wave function for $(\omega ,\omega ')$ and assume that $v(a_1,b_1)=d(a_1,b_1)$. Then $c=1$ in Theorem~\ref{thm21} so that
\begin{equation*}
v(a_j,b_j)=d(a_j,b_j)=\ab{a_j-b_j}
\end{equation*}
We do not lose generality with this assumption because we shall later normalize our wave functions so the theory is independent of a multiplicative constant. Since $v$ depends on $\omega$ and $\omega '$ we sometimes write
\begin{equation*}
v(a_n,b_n)=v(\omega ,\omega ',a_n,b_n)
\end{equation*}
and think of $v$ as the amplitude that the particles move along $\omega$ and $\omega '$ starting at $(a_1,b_1)$ and ending at $(a_n,b_n)$. According to the usual quantum formalism, the amplitude that the particles move to positions $(a_n,b_n)$ becomes
\begin{equation}         
\label{eq23}
\vhat (a_n,b_n)=\sum _{\omega ,\omega '}v(\omega ,\omega ',a_n,b_n)=\sum _{\omega ,\omega '}\ab{a_n-b_n}
\end{equation}
where $(\omega ,\omega ')$ are all pairs of noncrossing paths from $(a_1,b_1)$ to $(a_n,b_n)$. Letting $\tau$ be the number of such path pairs \eqref{eq23} gives
\begin{equation}         
\label{eq24}
\vhat (a_n,b_n)=\tau\ab{a_n-b_n}
\end{equation}

Suppose $h(a_1)=h(b_1)=p$ and $h(a_n)=h(b_n)=q$, $p<q$. We then have that
\begin{equation}         
\label{eq25}
\tau =s_{p+1}(x)\sqbrac{s_{p+1}(x)-1}s_{p+2}(x)\sqbrac{s_{p+2}(x)-1}\cdots s_q(x)\sqbrac{s_q(x)-1}
\end{equation}
By \eqref{eq24} and \eqref{eq25}, the position amplitude $\vhat (a_n,b_n)$ does not depend on $a_1$ and $b_1$ and only depends on
$a_n$ and $b_n$ through the term $\ab{a_n-b_n}$. We next normalize $\vhat (a_n,b_n)$ to obtain the \textit{position state} $\psi (a_n,b_n)$. Thus, if $S_q(x)=\brac{c_1,c_2,\ldots ,c_r}$ where the $c_j$ are vertex labels in increasing order, then the normalization constant becomes:
\begin{equation*}
N^2=\sum\brac{\ab{\vhat (c_j,c_k)}^2\colon c_j<c_k}
\end{equation*}
and $\psi (c_j,c_k)=\vhat (c_j,c_k)/N$. Since $\tau$ is a constant that appears in every term of $N$ and also in $\vhat (c_j,c_k)$ it cancels out and we can omit it. We then obtain the normalization constant
\begin{equation*}
N_1^2=\sum _{j=1}^{r-1}(r-j)j^2
\end{equation*}
and $\psi (c_j,c_k)=\ab{j-k}/N_1$.

As usual in the quantum formalism, we define the probability that the two particles end at positions $(c_j,c_k)$ to be
\begin{equation*}
p(c_j,c_k)=\ab{\psi (c_j,c_k)}^2=\frac{\ab{j-k}^2}{N_1^2}
\end{equation*}
We call $p(c_j,c_k)$ the \textit{force-free distribution}. The constant $N_1$ is not so important. What is important is the relative sizes of the probabilities $p(c_j,c_k)$. For example, if $c_k-c_j=r-1$ and $c_t-c_s=1$, then
\begin{equation*}
\frac{p(c_s,c_t)}{p(c_j,c_k)}=\frac{1}{(r-1)^2}
\end{equation*}
which can be quite small. However, we must remember that there are $r-1$ vertices with $c_t-c_s=1$ and only one vertex with $c_k-c_j=r-1$ so letting $A$ be the set of vertices satisfying $c_t-c_s=1$ and $B$ be the set satisfying $c_k-c_j=r-1$ we have
\begin{equation*}
\frac{\rmprob (A)}{\rmprob (B)}=\frac{1}{r-1}
\end{equation*}
This can still be small but it is not as small as $1/(r-1)^2$.

\begin{exam}{1}  
Suppose $S_q=\brac{c_1,c_2,c_3,c_4,c_5}$. We then have
\begin{equation*}
N_1^2=\sum _{j=1}^4(5-j)j^2=4+3\ctimes 2^2+2\ctimes 3^2+4^2=50
\end{equation*}
The probabilities become:
\begin{align*}
p(c_1,c_2)&=p(c_2,c_3)=p(c_3,c_4)=p(c_4,c_5)=\tfrac{1}{50}\\\noalign{\medskip}
p(c_1,c_3)&=p(c_2,c_4)=p(c_3,c_5)=\tfrac{4}{50}=\tfrac{2}{25}\\\noalign{\medskip}
p(c_1,c_4)&=p(c_2,c_5)=\tfrac{9}{50}\\\noalign{\medskip}
p(c_1,c_5)&=\tfrac{16}{50}=\tfrac{8}{25}
\end{align*}
We conclude that being three apart is more probable than any other configuration.\qquad\qedsymbol
\end{exam} 

\section{Electric Forces} 
We now discuss an electric force wave equation. As in Section~2, let $\omega =a_1a_2\cdots a_n$, $\omega '=b_1b_2\cdots b_n$ be two paths in $x$. If $\omega ,\omega '$ are equal-time paths, we assume they are noncrossing. Suppose we have two charged particles with electric charges $e_1$ and $e_2$. We again let $v$ be a complex-valued function of two variables whose domain includes
$\brac{(a_j,b_j)\colon j=1,2,\ldots ,n}$. If $e_1e_2\ge 0$, the \textit{repulsive electric-force wave equation} is defined to be
\begin{equation}         
\label{eq31}
i\nabla _{\omega ,\omega '}v(a_j,b_j)=\frac{e_1e_2}{d(a_j,b_j)}\,v(a_j,b_j)
\end{equation}
$j=2,3,\ldots ,n$. If $e_1e_2\le 0$, the \textit{attractive electric-force wave equation} is defined to be
\begin{equation}         
\label{eq32}
i\nabla _{\omega ,\omega '}v(a_j,b_j)=\frac{e_1e_2}{d(a_j,b_j)}\,v(a_{j-1},b_{j-1})
\end{equation}
$j=2,3,\ldots ,n$. Of course, if $e_1=0$ or $e_2=0$, then \eqref{eq31} and \eqref{eq32} reduce to the force-free wave equation \eqref{eq22}. If $v$ satisfies \eqref{eq31} or \eqref{eq32} for $j=2,3,\ldots ,n$, then $v$ is called an \textit{electric force wave function}. The next result shows that \eqref{eq31} or \eqref{eq32} and the \textit{initial condition} $v(a_1,b_1)$ uniquely determine $v$.

\begin{thm}       
\label{thm31}
A function $v(a_j,b_j)$, $j=1,2,\ldots ,n$, satisfies \eqref{eq31} if and only if 
\begin{equation}        
\label{eq33}
v(a_j,b_j)=\frac{i^{j_1}}{\sqbrac{\frac{id(a_{j-1},b_{j-1})}{d(a_j,b_j)}\,-\,\frac{e_1e_2}{d(a_j,b_j)^2}}\cdots
 \sqbrac{\frac{id(a_1,b_1)}{d(a_2,b_2)}\,-\,\frac{e_1e_2}{d(a_2,b_2)^2}}}\,v(a_1,b_1)
\end{equation}
and $v(a_j,b_j)$, $j=1,2,\ldots ,n$, satisfies \eqref{eq32} if and only if 
\begin{equation}        
\label{eq34}
v(a_j,b_j)=\frac{d(a_j,b_j)^2-ie_1e_2}{d(a_{j-1},b_{j-1})d(a_j,b_j)}\cdots
 \frac{d(a_2,b_2)^2-ie_1e_2}{d(a_1,b_1)d(a_2,b_2)}\,v(a_1,b_1)
\end{equation}
\end{thm}
\begin{proof}
If $v(a_j,b_j)$ satisfies \eqref{eq31}, we have
\begin{align*}  
i\sqbrac{d(a_{j-1},b_{j-1}v(a_j,b_j)-d(a_j,b_j)v(a_{j-1},b_{j-1}}=\frac{e_1e_2}{d(a_j,b_j)}\,v(a_j,b_j)\\
\intertext{Hence,}
v(a_j,b_j)=\frac{id(a_j,b_j)}{id(a_{j-1},b_{j-1})-\frac{e_1e_2}{d(a_j,b_j)}}\,v(a_{j-1},b_{j-1})
\end{align*}
Continuing this process we obtain
\begin{align*}
v(a_j,b_j)&=\frac{id(a_j,b_j)}{id(a_{j-1},b_{j-1})-\frac{e_1e_2}{d(a_j,b_j)}}
  \ctimes\frac{id(a_{j-1},b_{j-1})}{id(a_{j-2},b_{j-2})-\frac{e_1e_2}{d(a_{j-1},b_{j-1})}}\,v(a_{j-2},b_{j-2})\\
  &\ \vdots\\
  &=\frac{id(a_j,b_j)}{id(a_{j-1},b_{j-1})-\frac{e_1e_2}{d(a_j,b_j)}}\cdots\frac{id(a_2,b_2)}{id(a_1,b_1)-\frac{e_1e_2}{d(a_2,b_2)}}
  \,v(a_1,b_1)
\end{align*}
from which the result follows. If $v(a_j,b_j)$ satisfies \eqref{eq32} we have
\begin{equation*}
i\sqbrac{d(a_{j-1},b_{j-1}v(a_j,b_j)-d(a_j,b_j)v(a_{j-1},b_{j-1}}=\frac{e_1e_2}{d(a_j,b_j)}\,v(a_{j-1},b_{j-1})
\end{equation*}
Hence,
\begin{align*}
v(a_j,b_j)&=\frac{id(a_j,b_j)^2+e_1e_2}{id(a_{j-1},b_{j-1})d(a_j,b_j}\,v(a_{j-1},b_{j-1})\\\noalign{\medskip}
  &=\frac{d(a_j,b_j)^2-ie_1e_2}{d(a_{j-1},b_{j-1})d(a_j,b_j)}\,v(a_{j-1},b_{j-1})\\
    &\ \vdots\\
  &=\frac{d(a_j,b_j)^2-ie_1e_2}{d(a_{j-1},b_{j-1})d(a_j,b_j)}\cdots\frac{d(a_2,b_2)^2-ie_1e_2}{id(a_1,b_1)d(a_2,b_2)}\,v(a_1,b_1)
\end{align*}
The converses are straightforward.
\end{proof}

\begin{exam}{2}  
One of the simplest nontrivial $c$-causets has shell sequence $(1,2,2,3)$. Delineating shells by semicolons, we can label the vertices as $(1;2,3;4,5;6,7,8)$. We now examine two indistinguishable particles both having charge $e$ that move from the first shell $(2,3)$ to the top shell $(6,7,8)$. The paths between these shells are:
$\omega _1=2-4-6$, $\omega _2=3-4-6$, $\omega _3=2-5-6$, $\omega _4=3-5-6$, $\omega _5=2-4-7$, $\omega _6=2-5-7$,
$\omega _7=3-4-7$, $\omega _8=3-5-7$, $\omega _9=2-4-8$, $\omega _{10}=2-5-8$, $\omega _{11}=3-4-8$, $\omega _{12}=3-5-8$. Applying \eqref{eq33} we obtain
\begin{align*}
v(\omega _1,\omega _8,6,7)&=v(\omega _2,\omega _6,6,7)=v(\omega _3,\omega _7,6,7)=v(\omega _4,\omega _5,6,7)\\
  &=-\frac{1}{(i-e^2)^2}=\frac{(1-e^4)-2e^2i}{(1-e^4)^2+4e^4}
\end{align*}
Moreover,
\begin{align*}
v(\omega _1,\omega _{12},6,8)&=v(\omega _1,\omega _{10},6,8)=v(\omega _3,\omega _{11},6,8)=v(\omega _8,\omega _9,6,8)\\
  &=-\frac{-1}{\paren{\frac{i}{2}-\frac{e^2}{4}}(i-e^2)}=\frac{4(2-e^4)-12e^2i}{(2-e^4)^2+9e^4}
\end{align*}
Summing over noncrossing path pairs as in \eqref{eq23} we obtain the position wave functions
\begin{align*}
\vhat (6,7)&=4\sqbrac{\frac{(1-e^4)-2e^2i}{(1-e^4)^2+4e^4}}\\\noalign{\medskip}
\vhat (6,8)&=4\sqbrac{\frac{4(2-e^4)-12e^2i}{(2-e^4)^2+9e^4}}
\end{align*}
and by symmetry $\vhat (7,8)=\vhat (6,7)$. The normalization constant becomes
\begin{equation*}
N^2=2\ab{\vhat (6,7)}^2+\ab{\vhat (6,8)}^2=\frac{32}{(1+e^4)^2}+\frac{256}{(1+e^4)(4+e^4)}
\end{equation*}
We compute the probabilities as
\begin{align*}
p(6,7)=p(7,8)=\frac{\ab{\vhat (6,7)}^2}{N^2}=\frac{4+e^4}{6(4+3e^4)}\\
\intertext{and}
p(6,8)=1-2p(6,7)=\frac{8(1+e^4)}{3(4+3e^4)}
\end{align*}
Now $p(6,7)$ and $p(6,8)$ are functions of $e$ and we write $p(6,7)=p(6,7)(e)$ and $p(6,8)=p(6,8)(e)$. Notice that $p(6,7)(0)=1/6$ and
$p(6,8)(0)=2/3$ which agree with the force-free probabilities of Section ~2. Moreover, for large $e$ we have
\begin{equation*}
\lim _{e\to\infty}p(6,7)(e)=\tfrac{1}{18},\quad\lim _{e\to\infty}p(6,8)(e)=\tfrac{8}{9}
\end{equation*}
Thus, under this repulsive electric force, the particles prefer to be farther apart. The function $p(6,7)(e)$ decreases monotonically from $1/6$ at $e=0$ to $1/18$ at $e\approx 3$. The function $p(6,8)(e)$ increases monotonically from $2/3$ at $e=0$ to $8/9$ at $e\approx 3$.
\qquad\qedsymbol
\end{exam} 

Although we shall not use this in the sequel, we now briefly discuss how to place this theory within the tensor product formalism. In Example~2, we think of $S_3(x)=\brac{6,7,8}$ as a 3-dimensional Hilbert space with orthonormal basis
$\sixhat ,\sevenhat ,\eighthat$. We can extend the states $\psi (j,k)=\vhat (j,k)/N$ by defining $\psi (k,j)=-\psi(j,k)$. In fact, we see from
\eqref{eq33} and \eqref{eq34} that if $v(a_1,b_1)=-v(b_1,a_1)$ then $v(a_j,b_j)$ and hence $\psi (a_j,b_j)$ are automatically antisymmetric. (We could also make them symmetric in this way.) If we do this antisymmetrization, the state we have constructed in Example~2 can be written as
\begin{align*}
\psi&=\frac{1}{\sqrt{2}}\left[\psi (6,7)\sixhat\otimes\sevenhat +\psi (7,6)\sevenhat\otimes\sixhat +\psi (7,8)\sevenhat\otimes\eighthat
  +\psi(8,7)\eighthat\otimes\sevenhat \right.\\
  &\quad\left.+\psi (6,8)\sixhat\otimes\eighthat +\phi (8,6)\eighthat\otimes\sixhat\right]\\
  &=\psi (6,7)\frac{(\,\sixhat\otimes\sevenhat -\sevenhat\otimes\sixhat\,)}{\sqrt{2}}
  +\psi (7,8)\frac{(\,\sevenhat\otimes\eighthat -\eighthat\otimes\sevenhat\,)}{\sqrt{2}}
  +\psi (6,8)\frac{(\,\sixhat\otimes\eighthat -\eighthat\otimes\sixhat\,)}{\sqrt{2}}
\end{align*}
Then the probability that the particle starting at vertex~2 ends at vertex~6 while the particle starting at vertex~3 ends at vertex~7 is
\begin{equation*}
\ab{\elbows{\psi ,\sixhat\otimes\sevenhat}}^2=\tfrac{1}{2}\ab{\psi (6,7)}^2
\end{equation*}
Similarly, if we reverse vertices 2 and 3 in the above we obtain
\begin{equation*}
\ab{\elbows{\psi ,\sevenhat\otimes\sixhat}}^2=\tfrac{1}{2}\ab{\psi (6,7)}^2
\end{equation*}
Similar results hold for vertex pairs $(7,8)$ and $(6,8)$.

\begin{exam}{3}  
We can adjoin a vertex labeled 9 to the top shell in Example~2 to obtain the shell sequence $(1,2,2,4)$. This $c$-causet can be analyzed by constructing the new paths $\omega _{13}=2-4-9$, $\omega _{14}=2-5-9$, $\omega _{15}=3-4-9$, $\omega _{16}=3-5-9$. We then have 
\begin{align*}
v(\omega _1,\omega _{16},6,9)&=v(\omega _2,\omega _{14},6,9)=v(\omega _3,\omega _{15},6,9)=v(\omega _4,\omega _{13},6,9)\\
  &=-\frac{-1}{\paren{\frac{i}{3}-\frac{e^2}{9}}(i-e^2)}=\frac{9(3-e^4)-36e^2i}{(3-e^4)^2+16e^4}
\end{align*}
This gives
\begin{align*}
\vhat (6,9)=36\sqbrac{\frac{(3-e^4)-4e^2i}{(3-e^4)^2+16e^4}}\\
\intertext{Hence,}
\ab{\vhat (6,9)}^2=\frac{36^2}{(9+e^4)(1+e^4)}
\end{align*}
Using the values from Example~2 we obtain
\begin{align*}
p(6,7)&=\frac{\ab{\vhat (6,7)}^2}{3\ab{\vhat (6,7)}^2+2\ab{\vhat (6,8)}^2+\ab{\vhat (6,9)}^2}\\
&=\frac{1}{3+2\frac{\ab{\vhat (6,8)}^2}{\ab{\vhat (6,7)}^2}+\frac{\ab{\vhat (6,9)}^2}{\ab{\vhat (6,7)}^2}}\\
&=\frac{1}{3+\frac{32(1+e^4)}{4+e^4}+\frac{81(1+e^4)}{9+e^4}}
\end{align*}
In a similar way we obtain
\begin{align*}
p(6,8)&=\frac{1}{2+\frac{3(4+e^4)}{16(1+e^4)}+\frac{81(4+e^4)}{16(9+e^4)}}\\
p(6,9)&=\frac{1}{1+\frac{9+e^4}{27(1+e^4)}+\frac{32(9+e^4)}{81(4+e^4)}}
\end{align*}
As in Section~2 we obtain the force-free distribution
\begin{align*}
p(6,7)(0)&=p(7,8)(0)=p(8,9)(0)=\tfrac{1}{20}\\
p(6,8)(0)&=p(7,9)(0)=\tfrac{1}{5}\\
p(6,9)(0)&=\tfrac{9}{20}
\end{align*}
For large values of $e$ (even for $e\approx 3$) we obtain
\begin{equation*}
\lim _{e\to\infty}p(6,7)(e)=\tfrac{1}{116},\quad\lim _{e\to\infty}p(6,8)(e)=\tfrac{16}{116},\quad\lim _{e\to\infty}p(6,9)(e)=\tfrac{81}{116}
\end{equation*}
The functions $p(6,7)(e)$ and $p(6,8)(e)$ decrease monotonically, while $P(6,9)(e)$ increases monotonically. Again, the two particles prefer to move farther apart. This phenomenon becomes more pronounced as more vertices are adjoined to the top shell.\qquad\qedsymbol
\end{exam} 

\begin{exam}{4}  
This example considers the attractive wave equation \eqref{eq32} with $e_1=-e_2=e$ for the $c$-causet of Example~2. Applying \eqref{eq34} we have
\begin{align*}
v(\omega _1,\omega _8,6,7)&=v(\omega _2,\omega _6,6,7)=v(\omega _3,\omega _7,6,7)=v(\omega _4,\omega _5,6,7)\\
&=(1+ie^2)^2\\
v(\omega _1,\omega _{12},6,8)&=v(\omega _2,\omega _{10},6,8)=v(\omega _3,\omega _{11},6,8)=v(\omega _4,\omega _9,6,8)\\
&=\tfrac{1}{2}(4+ie^2)(1+ie^2)
\end{align*}
We conclude that
\begin{align*}
\vhat (6,7)&=\vhat (7,8)=4(1-e^4)+8e^2i\\
\vhat (6,8)&=2(4-e^4)+10e^2i
\end{align*}
Hence,
\begin{align*}
\ab{\vhat (6,7)}^2&=\ab{\vhat (7,8)}^2=16(1+e^4)^2\\
\ab{\vhat (6,8)}^2&=4(1+e^4)(16+e^4)
\end{align*}
Computing the probabilities gives
\begin{align*}
p(6,7)&=p(7,8)=\frac{\ab{\vhat (6,7)}^2}{2\ab{\vhat (6,7)}^2+\ab{\vhat (6,8)}^2}=\frac{1}{2+\frac{\ab{\vhat (6,8)}^2}{\ab{\vhat (6,7)}^2}}\\
  &=\frac{4(1+e^4)}{3(8+3e^4)}\\
  p(6,8)&=\frac{1}{1+\frac{2\ab{\vhat (6,7)}^2}{\ab{\vhat (6,8)}^2}}=\frac{16+e^4}{3(8+3e^4)}
\end{align*}
Of course, we also have that $p(6,8)=1-2(p(6,7)$. As before, we obtain
\begin{equation*}
p(6,7)(0)=p(7,8)(0)=1/16,\quad p(6,8)(0)=2/3
\end{equation*}
However, in contrast with Examples~2 and 3 we have
\begin{align*}
\lim _{e\to\infty}p(6,7)(e)&=\lim _{e\to\infty}p(7,8)(e)=\tfrac{4}{9}\\
\lim _{e\to\infty}p(6,8)(e)&=\tfrac{1}{9}
\end{align*}
Thus, the particles prefer to be close together under the attractive force. Notice that
\begin{equation*}
p(6,7)(1)=p(7,8)(1)=\tfrac{8}{33},\quad p(6,8)(1)=\tfrac{17}{33}
\end{equation*}
so that for $e\approx 1$, the probabilities that the two particles are within one unit and two units are about the same.\qquad\qedsymbol
\end{exam} 

\begin{exam}{5}  
We now compute probabilities for the $c$-causet in Example~3 with an attractive electric force as in Example~4. We again have the new paths $\omega _{13},\omega _{14},\omega _{15},\omega _{16}$ for which we obtain
\begin{align*}
v(\omega _1,\omega _{16},6,9)&=v(\omega _2,\omega _{14},6,9)=v(\omega _3,\omega _{15},6,9)=v(\omega _4,\omega _{13},6,9)\\
  &=\frac{9+ie^2}{3}\,(1+ie^2)
\end{align*}
This gives
\begin{align*}
\vhat (6,9)=\tfrac{4}{3}(9-e^4+10e^2i)\\
\intertext{Hence,}
\ab{\vhat (6,9)}^2=\tfrac{16}{9}\,(1+e^4)(81+e^4)
\end{align*}
Using the values from Example~4 we obtain
\begin{align*}
p(6,7)&=\frac{1}{3+\frac{2\ab{\vhat (6,8)}^2}{\ab{\vhat (6,7)}^2}+\frac{\ab{\vhat (6,9)}^2}{\ab{\vhat (6,7)}^2}}
  =\frac{1}{3+\frac{16+e^4}{2(1+e^4)}+\frac{81+e^4}{9(1+e^4)}}\\
  &=\frac{18(1+e^4)}{5(72+13e^4)}
\end{align*}
In a similar way we have
\begin{align*}
p(6,8)&=\frac{9(16+e^4)}{10(72+13e^4)}\\
p(6,9)&=\frac{2(81+e^4)}{5(72+13e^4)}
\end{align*}
As before
\begin{equation*}
p(6,7)(0)=\tfrac{1}{20},\quad p(6,8)(0)=\tfrac{1}{5},\quad p(6,9)(0)=\frac{9}{20}
\end{equation*}
But for large $e$ we have
\begin{equation*}
\lim _{e\to\infty}p(6,7)(e)=\tfrac{18}{65},\quad\lim _{e\to\infty}p(6,8)(e)=\tfrac{9}{130},\quad\lim _{e\to\infty}p(6,9)(e)=\tfrac{2}{65}
\end{equation*}
Again, for $e\approx 3$ and larger, the particles prefer to be together.\qquad\qedsymbol
\end{exam} 

\section{Strong Nuclear Force} 
We propose that the strong nuclear force acts like an attractive spring force with wave equation given by
\begin{equation}         
\label{eq41}
i\nabla _{\omega ,\omega '}v(a_j,b_j)=-gd(a_{j-1},b_{j-1})v(a_{j-1},b_{j-1})
\end{equation}
where $g>0$ is a constant. As in Theorem~\ref{thm31}, a function $v(a_j,b_j)$ satisfies \eqref{eq41} for $j=2,3,\ldots ,n$ if and only if $v$ as the form
\begin{equation}         
\label{eq42}
v(a_j,b_j)=\sqbrac{\frac{d(a_j,b_j)}{d(a_{j-1},b_{j-1})}+ig}\cdots\sqbrac{\frac{d(a_2,b_2)}{d(a_1,b_1)}+ig}v(a_1,b_1)
\end{equation}
We now repeat Examples~2 and 3 for the strong nuclear force.

\begin{exam}{6}  
As in Example~2, we consider the $c$-causet with shell sequence $(1,2,2,3)$. Applying \eqref{eq42} we obtain
\begin{align*}
v(\omega _1,\omega _8,6,7)&=v(\omega _2,\omega _6,6,7)=v(\omega _3,\omega _7,6,7)=v(\omega _4,\omega _5,6,7)\\
  &=(1+ig)^2\\
v(\omega _1,\omega _{12},6,8)&=v(\omega _1,\omega _{10},6,8)=v(\omega _3,\omega _{11},6,8)=v(\omega _8,\omega _9,6,8)\\
  &=(2+ig)(1+ig)
\end{align*}
Hence,
\begin{align*}
\vhat (6,7)&=4(1+ig)^2=4(1-g^2)+8gi\\
\vhat (6,8)&=4(2+ig)(1+ig)=4(2-g^2)+12gi
\end{align*}
It follows that
\begin{equation*}
\ab{\vhat (6,7)}^2=16(1+g^2)^2,\quad\ab{\vhat (6,8)}^2=16(1+g^2)(4+g^2)
\end{equation*}
We then obtain
\begin{align*}
p(6,7)&=p(7,8)=\frac{1}{2+\frac{\ab{\vhat (6,8)}^2}{\ab{\vhat (6,7)}^2}}=\frac{1}{2+\frac{4+g^2}{1+g^2}}=\frac{1+g^2}{3(2+g^2)}\\
p(6,8)&=1-2p(6,7)=\frac{4+g^2}{3(2+g^2)}
\end{align*}
For the larger $g$ we have
\begin{equation*}
\lim _{g\to\infty}p(6,7)(g)=\lim _{g\to\infty}p(7,8)(g)=\lim _{g\to\infty}p(6,8)(g)=\tfrac{1}{3}
\end{equation*}
Since $g$ is a constant, we should not let $g$ vary. However, if $g$ about 3 or higher, then $p(6,7)$, $p(7,8)$ and $p(6,8)$ are all close to $1/3$.\qquad\qedsymbol
\end{exam} 

\begin{exam}{7}  
We next consider the strong nuclear force for the $c$-causet with shell sequence $(1,2,2,4)$. As in Example~3 we have the new paths
$\omega _{13},\omega _{14}, \omega _{15},\omega _{16}$. Each of the noncrossing path pairs to $(6,9)$ have the wave function
\begin{align*}
v(\omega _1,\omega _{16},6,9)&=(3+ig)(1+ig)\\
\intertext{Hence,}
\vhat (6,9)&=4(3+ig)(1+ig)=4(3-g^2)+16gi\\
\intertext{and}
\ab{\vhat (6,9)}^2&=16(1+g^2)(9+g^2)
\end{align*}
We conclude that
\begin{align*}
p(6,7)&=\frac{1}{3+\frac{2\ab{\vhat (6,8)}^2}{\ab{\vhat (6,7)}^2}+\frac{\ab{\vhat (6,9)}^2}{\ab{\vhat (6,7)}^2}}
  =\frac{1}{3+\frac{2(4+g^2)}{1+g^2}+\frac{9+g^2}{1+g^2}}\\
  &=\frac{1+g^2}{2(10+3g^2)}
\end{align*}
In a similar way
\begin{equation*}
p(6,8)=\frac{4+g^2}{2(10+3g^2)}\,,\quad p(6,9)=\frac{9+g^2}{2(10+3g^2)}
\end{equation*}
We then obtain
\begin{equation*}
\lim _{g\to\infty}p(6,7)(g)=\lim _{g\to\infty}p(6,8)(g)=\lim _{g\to\infty}p(6,9)(g)=\tfrac{1}{6}
\end{equation*}
which is a bit surprising.\qquad\qedsymbol
\end{exam} 

\begin{exam}{8}  
The previous example compels us to consider the strong nuclear force for the $c$-causet with shell sequence $(1,2,2,5)$ to see if there is a pattern. We now have the new paths $\omega _{17}=2-4-10$, $\omega _{18}=2-5-10$, $\omega _{19}=3-4-10$,
$\omega _{20}=3-5-10$. Each of the noncrossing path pairs to $(6,10)$ have the wave function
\begin{equation*}
v(\omega _2,\omega _{17},6,10)=(4+ig)(1+ig)
\end{equation*}
Hence,
\begin{align*}
\vhat (6,10)&=4(4+ig)(1+ig)=4(4-g^2)+20gi\\
\intertext{and}
\ab{\vhat (6,10)}^2&=16(1+g^2)(16+g^2)
\end{align*}
We conclude that
\begin{align*}
p(6,7)&=\frac{\ab{\vhat (6,7)}^2}{4\ab{\vhat (6,7)}^2+3\ab{\vhat (6,8)}^2+2\ab{\vhat (6,9)}^2+\ab{\vhat (6,10)}^2}\\\noalign{\medskip}
  &=\frac{1}
  {4+\frac{3\ab{\vhat (6,8)}^2}{\ab{\vhat (6,7)}^2}+\frac{2\ab{\vhat (6,9)}^2}{\ab{\vhat (6,7)}^2}+\frac{\ab{\vhat (6,10)}^2}{\ab{\vhat (6,7)}^2}}
  \\\noalign{\medskip}
  &=\frac{1}{4+\frac{3(4+g^2)}{1+g^2}+\frac{2(9+g^2)}{1+g^2}+\frac{16+g^2}{1+g^2}}=\frac{1+g^2}{10(5+g^2)}
\end{align*}
In a similar way we obtain
\begin{equation*}
p(6,8)=\frac{4+g^2}{10(5+g^2)}\,,\quad p(6,9)=\frac{9+g^2}{10(5+g^2)}\,,\quad p(6,10)=\frac{16+g^2}{10(5+g^2)}
\end{equation*}
We then have
\begin{equation*}
p(6,7)(0)=\tfrac{1}{50},\quad p(6,8)(0)=\tfrac{2}{25},\quad p(6,9)(0)=\tfrac{9}{50},\quad (6,10)(0)=\tfrac{8}{25}
\end{equation*}
as predicted in Section~2. Moreover, we obtain
\begin{equation*}
\lim _{g\to\infty}p(6,7)(g)=\lim _{g\to\infty}p(6,8)(g)=\lim _{g\to\infty}p(6,9)(g)=\lim _{g\to\infty}p(6,10)(g)=\tfrac{1}{10}
\end{equation*}
It is not hard to show that this pattern continues and for a $c$-causet with shell sequence $(1,2,2,n)$ we have
\begin{equation*}
\lim _{g\to\infty}p(6,7)(g)=\lim _{g\to\infty}p(6,8)(g)=\cdots =\lim _{g\to\infty}p(6,n+5)(g)=\tfrac{n(n-1)}{2}
\end{equation*}
We suspect that this strong regularity is due to the fact that we are considering small $c$-causets. For larger $c$-causets we believe that the situation will be much more complex.\qquad\qedsymbol
\end{exam} 

\end{document}